\documentclass[12pt]{article}      
\usepackage{amsmath,amssymb,amsthm, amscd, graphicx, verbatim, setspace} 
\usepackage{setspace}

\theoremstyle{plain}
\newtheorem{thm}{Theorem}
\newtheorem{lem}[thm]{Lemma}

\title{Distributed pursuit algorithms for probabilistic adversaries on connected graphs}
\date{}
\author{Jesse Geneson\\
\small\tt geneson@gmail.com\\
}
\begin{document}
\maketitle

\begin{abstract}
A gambler moves between the vertices $1, \ldots, n$ of a graph using the probability distribution $p_{1}, \ldots, p_{n}$. Multiple cops pursue the gambler on the graph, only being able to move between adjacent vertices. We investigate the expected capture time for the gambler against $k$ cops as a function of $n$ and $k$ for three versions of the game:
\begin{itemize}
\item known gambler: the cops know the gambler's distribution
\item unknown gambler: the cops do not know the gambler's distribution
\item known changing gambler: the gambler's distribution can change every turn, but the cops know all of the gambler's distributions from the beginning
\end{itemize}
We show for $n > k$ that if the cops are allowed to choose their initial positions before the game starts and before they know the gambler's distribution(s), and if both the gambler and the cops play optimally, then the expected capture time is $\Theta(n/k)$ for the known gambler, the unknown gambler, and the known changing gambler.
\end{abstract}

\section{Introduction}

Graph pursuit games have been investigated for decades \cite{KW, T1, T2, T3, T4, NW, Q, G} because of their applications to anti-incursion programs. The adversary in the game is called a gambler if they move among the vertices $1, \ldots, n$ following a probability distribution $p_{1}, \ldots, p_{n}$. The gambler is called a known or unknown gambler depending on whether the cops know their probability distribution or not.

Gambler-pursuit games model anti-incursion programs navigating a linked list of ports, trying to minimize interception time for enemy packets. Komarov and Winkler proved that the expected capture time on any connected $n$-vertex graph is exactly $n$ for a known gambler against a single cop, assuming that both players use optimal strategies \cite{KW}. For an unknown gambler, Komarov and Winkler proved an upper bound of approximately $1.97n$ for a single cop \cite{KW}.

Komarov and Winkler conjectured that the general upper bound for the unknown gambler against a single cop on a connected $n$-vertex graph can be improved from about $1.97n$ to $3n/2$, and that the star is the worst case for this bound. We improved the upper bound for the unknown gambler against a single cop to $1.95n$ in \cite{G}. 

In Section \ref{2}, we prove a lemma about partitioning trees that we use for the distributed pursuit algorithms in the fourth and fifth sections. In Section \ref{3}, we show that if the cops can choose their initial positions before the game starts but after they learn the gambler's distribution, then the expected capture time on a connected $n$-vertex graph is $max(1,n/k)$ for the known gambler versus $k$ cops. In Section \ref{4}, we show an upper bound of $3.94n/k+1.16$ on the expected capture time of the unknown gambler versus $k$ cops on a connected $n$-vertex graph with $n > k$. 

In the last section, we consider a gambler that can change their probability distribution every turn, which we call the \emph{changing gambler}. When the cops know all of their distributions at the beginning, they are called a \emph{known changing gambler}. We show for $n > k$ that the expected capture time for the known changing gambler against $k$ cops on a connected graph with $n$ vertices is at most $6.33 n / k+3.17$, assuming that the cops choose their initial positions before the game starts and before they know the gambler's distributions.

\section{Choosing sectors for each cop}\label{2}

In the distributed pursuit algorithms for the unknown gambler in Section \ref{4} and the known changing gambler in Section \ref{5}, the cops split the $n$-vertex connected graph into sectors. There are at most $k$ sectors of size at most $2n / k+1$ and each cop goes to a sector. If there are fewer than $k$ sectors, then some cops will patrol the same sector. 

Next we explain how to choose the sectors for a given connected graph. Suppose the connected graph $G$ has $n$ vertices. Let $T$ be a spanning subtree of $G$. The sectors are chosen by repeatedly applying the following lemma with $x = n / k+1$. The sector is $T|_{S}$ and the vertex $v$ is the only vertex in the sector that might be included in other sectors.

\begin{lem}\label{parti}
If $T = (V, E)$ is a tree with $|V| = n$ and $n > x$, then there exists a subset $S \subset V$ and a vertex $v \in S$ such that $x < |S| \leq 2x-1$, $T|_{S}$ is connected, and $T|_{(V-S)\cup \left\{v \right\}}$ is connected.
\end{lem}

\begin{proof}
We prove this by induction on $n$. If $x < n \leq 2x-1$, then let $S = V$ and let $v$ be any vertex in $S$. If $n > 2x-1$, then suppose that the lemma is true for all $m$ such that $m < n$ and pick any vertex $u \in V$. 

Find the branch $B$ of $T$ starting from (and including) $u$ with the most vertices. If $|B| > 2x-1$, then use the inductive hypothesis on $B-\left\{u \right\}$ to find $S$ and $v$. If $x < |B| \leq 2x-1$, then let $S = B$ and let $v = u$. 

Otherwise if $|B| \leq x$, then let $C = B$ and continue to add all of the vertices from each next largest branch to $C$ until $C$ has more than $x$ vertices. Since $B$ was the largest branch starting from $u$, $x < |C| \leq 2x-1$ at the end of the process. Thus let $S = C$ and $v = u$.
\end{proof}

\section{Known gambler}\label{3}

If the cops start from the same location and the gambler knows the cops' location before choosing their distribution, then the maximum possible expected capture time for a connected $n$-vertex graph is $n$, regardless of the number of cops. 

\begin{lem}
If $k$ police that cannot choose their starting positions pursue a known gambler on a connected $n$-vertex graph, then the maximum possible expected capture time is $n$.
\end{lem}

\begin{proof}
The upper bound follows from the upper bound for $1$ cop in \cite{KW}, while the lower bound follows from considering the path $P_{n}$ with the cops starting at one end and the gambler being at the other end with probability $1$. 
\end{proof}

Now we consider the case when the cops can choose their initial positions knowing the gambler's distribution. The next lemma shows that the expected capture time is at least $max(1, n / k)$ assuming that the cops play optimally.

\begin{lem}\label{uni}
If the gambler uses the uniform distribution on a connected $n$-vertex graph, then the expected capture time is $max(1, n / k)$.
\end{lem}

\begin{proof}
If there are $k < n$ cops that start at different positions, then the expected capture time is $n / k$ if the gambler uses the uniform distribution. If $k \geq n$, then there is a cop at every vertex and the gambler is caught on the first turn.
\end{proof}

Next we show that the expected capture time for the known gambler is at most $max(1, n / k)$, regardless of the distribution that the gambler uses, assuming that the cops know their distribution before choosing their starting locations.

\begin{lem}
If the cops know the gambler's distribution before the game and stay at the $k$ most probable vertices, then the expected capture time is at most $max(1, n / k)$.
\end{lem}

\begin{proof}
If the cops stay at the $k$ most probable vertices, then on any turn the probability that the cops catch the gambler is at least $min(1, k / n)$. Thus the expected capture time is at most $max(1, n / k)$.
\end{proof} 

In the next section, we consider what happens when the cops do not know the gambler's distribution. In this case, we are still able to get an $O(n/k)$ bound on the expected capture time.

\section{Unknown gambler}\label{4}

Assume that the $k$ cops are allowed to start from any initial positions, but they do not know the gambler's distribution. We show that the expected capture time for the unknown gambler on an $n$-vertex graph is at most $3.94n / k + 1.16$. This pursuit algorithm for $k$ cops is a distributed version of the pursuit algorithm for $1$ cop in \cite{KW}.

\begin{thm}
The expected capture time for the unknown gambler versus $k$ cops on a connected $n$-vertex graph is at most $3.94n / k+1.16$ for $n > k$.
\end{thm}

\begin{proof}
Let $G$ be the connected $n$-vertex graph. Use Lemma \ref{parti} to partition $G$ into at most $k$ sectors that are connected subgraphs $G_{1}, \ldots, G_{k}$ of size at most $2n / k+1$ and let $T_{1}, \ldots, T_{k}$ be spanning subtrees of $G_{1}, \ldots, G_{k}$ respectively. Suppose that for each $1 \leq i \leq k$, cop $i$ performs a depth first search of $T_{i}$, but they stay at the leaves of $T_{i}$ for an extra turn.

The cops flip coins to decide whether to perform their depth first searches forward or backward. Thus the total number of turns in a single depth first search (including the extra turns for the leaves) is at most $6 n / k+1$. The search is repeated until capture. Since the cops flip coins to decide whether to search forward or backward, the expected number of turns in the successful depth first search is at most $3 n / k+1$. Let one round of the pursuit algorithm refer to the cops completing a single depth first search in parallel, with some cops sitting during the end of the round to wait for the other cops to finish their round.

As in \cite{KW}, suppose that the gambler uses probability distribution $p_{1}, \ldots, p_{n}$. Then the probability of evasion for a single round is at most $\prod_{i=1}^{n} (1-p_{i})^{2} < e^{-2}$, so an average of $1/(1-e^{-2})$ rounds will suffice for capture. Thus the expected capture time is at most $ (1/(1-e^{-2})-1) (6n / k + 1) + 3  n / k + 1 \approx 3.94n / k + 1.16$.
\end{proof}

\section{Changing gamblers}\label{5}

The expected capture time for a changing gambler is $\Omega(n/k)$ by Lemma \ref{uni}. We show below that the known changing gambler has $O(n/k)$ expected capture time. 

\begin{lem}
For $n > k$ the expected capture time for the known changing gambler against $k$ cops on a connected graph with $n$ vertices is at most $6.33 n / k+3.17$, assuming that the cops choose their initial positions before the game starts and before they know the gambler's distributions.
\end{lem}

\begin{proof}
Let $G$ be the connected $n$-vertex graph. Using Lemma \ref{parti}, partition $G$ into at most $k$ sectors that are connected subgraphs $G_{1}, \ldots, G_{k}$ of size at most $2n / k+1$ and let $T_{1}, \ldots, T_{k}$ be spanning subtrees of $G_{1}, \ldots, G_{k}$ respectively. Before the game starts, each cop goes to an initial position in each sector. 

Each round of the pursuit algorithms has $2$ parts consisting of $2n/k+1$ turns each: in the first part, the cops walk to the vertices in their sector that have the highest average probability for the turns in the second part of the round and sit at the vertices once they get to them; in the second part, the cops sit on the vertices with the highest average probability in their sector. 

If there are any sectors with a shared vertex that has the highest average probability in multiple sectors, then one cop goes to that vertex and the other cop(s) in the sectors go to the vertices with the next highest average probability in their sectors.

In one round, the probability of evasion for the gambler is at most $((1-p_1)..(1-p_k))^{2n/k+1}$ for some $p_{1}, \ldots, p_{k} \in [0,1]$ such that $p_1+\ldots+p_k \geq \frac{1}{2n/k+1}$. Thus the probability of evasion is at most $((1-\frac{1}{2n+k})^k)^{2n/k+1} \leq 1/e$ and the expected capture time is at most $(4n/k+2)/(1-1/e) < 6.33 n / k+3.17$.
\end{proof}

The result below gives a better upper bound for stars and other graphs of diameter at most $6$.

\begin{lem}
The known changing gambler has expected capture time at most $1+d n / k$ on a connected $n$-vertex graph of diameter $d$ against $k$ cops for $n > k$.
\end{lem}

\begin{proof}
For each $i \geq 0$, suppose that the $k$ cops identify the $k$ vertices $v_{1}, \ldots, v_{k}$ that have the highest probabilities on turn $d (i+1)+1$. During turns $d i + 2$ through $d i + d$, the cops move to $v_{1}, \ldots, v_{k}$. If they reach $v_{1}, \ldots, v_{k}$, they sit at $v_{1}, \ldots, v_{k}$. On turn $d (i+1)$+1, the cops sit at $v_{1}, \ldots, v_{k}$. For each $i \geq 0$, let the turns $d i + 2$ through $d (i+1)+1$ be considered one round. 

The probability of capture for the gambler in a single round is at least $k / n$, so the expected number of rounds is at most $n / k$. Thus the expected capture time is at most $1+d n/k$.
\end{proof}

We end with a few open questions about the expected capture time for different kinds of gamblers, assuming that the cops choose their initial positions before they know the gamblers' distribution(s):

\begin{enumerate}

\item What is the maximum possible expected capture time for the known gambler versus $k$ cops on a connected $n$-vertex graph?

\item What is the maximum possible expected capture time for the unknown gambler versus $k$ cops on a connected $n$-vertex graph?

\item What is the maximum possible expected capture time for the known changing gambler versus $k$ cops on a connected $n$-vertex graph?

\item What is the maximum possible expected capture time for the unknown changing gambler versus $k$ cops on a connected $n$-vertex graph?

\end{enumerate}


\begin{thebibliography}{7}
\bibitem{T1} A. Berarducci and B. Intrigila, On the cop number of a graph, Adv. Appl. Math. 14 (1993), 389-403.
\bibitem{T2} A. Bonato, P. Golovach, G. Hahn, and J. Kratochvil, The capture time of a graph, Discrete Math. 309 (2009), 5588-5595.
\bibitem{G} J. Geneson, An anti-incursion algorithm for unknown probability distributions on a connected graph, manuscript (2016)
\bibitem{T3} G. Hahn, F. Laviolette, N. Sauer, and R.E. Woodrow, On cop-win graphs, Discrete Math. 258 (2002), 27-41.
\bibitem{T4} G. Hahn and G. MacGillivray, A note on k-cop, l robber games on graphs, Discrete Math. 306 (2006), 2492-2497.
\bibitem{KW} N. Komarov and P. Winkler, Cop vs. Gambler, Discrete Math. 339 (2016), 1677-1681.
\bibitem{NW} R. Nowakowski and P. Winkler, Vertex to vertex pursuit in a graph, Discrete Math. 43 (1983), 235-239.
\bibitem{Q} A. Quilliot, Homomorphismes, points fixes, r\'{e}tractations et jeux de poursuite dans les graphes, les ensembles ordonn\'{e}s et les espaces m\'{e}triques, Ph.D. thesis, Universit\'{e} de Paris VI., 1983.
\end{thebibliography}
\end{document}